\documentclass{llncs}

\usepackage{hyperref}
\usepackage{amssymb}
\usepackage{amsmath}
\usepackage{color}
\usepackage{epsfig}
\usepackage{thm-restate}

\usepackage[cmyk,dvipsnames]{xcolor}

\spnewtheorem{obs}{Observation}{\bfseries}{\itshape}
\DeclareMathOperator{\emw}{emw}
\DeclareMathOperator{\tw}{tw}

\begin{document}
\title{Embedded-width: A variant of treewidth for plane graphs}
\author{Glencora Borradaile\inst{1} \and Jeff Erickson\inst{2} \and Hung Le\inst{1} \and Robbie Weber\inst{3}}
\institute{Oregon State University \and University of Illinois, Urbana-Champaign \and University of Washington}
\maketitle
\begin{abstract}
We define a special case of tree decompositions for planar graphs  that respect a given embedding of the graph.  We study the analogous width of the resulting decomposition we call the {\em embedded-width} of a plane graph.  We show both upper bounds and lower bounds for the embedded-width of a graph in terms of its treewidth and describe a fixed parameter tractable algorithm to calculate the embedded-width of a plane graph.  To do so, we give novel bounds on the size of matchings in planar graphs.
\end{abstract}

\section{Introduction}

Treewidth is a graph parameter introduced by Robertson and Seymour \cite{RobertsonGM2}, measuring how ``tree-like'' a graph is.  (Formal definitions are in Section~\ref{sec:twDef}.)
	Treewidth has proven useful both in the development of theoretical results such as The Graph Minor Theorem, and in the design of algorithms on NP-hard problems~\cite{Bodlaender94}. Algorithmically, tree decompositions demonstrate a structure in a graph that can be exploited by dynamic programming approaches~\cite{Bodlaender88Dynamic}, which is particularly useful for planar graphs~\cite{Baker94}.
Calculating the treewidth of an arbitrary graph is NP-hard~\cite{Arnborg87}, and it is a long-standing open problem whether restricting the input to planar graphs allows for a polynomial-time algorithm or whether it remains NP-hard~\cite{Bodlaender94}.

To better understand the relationship between treewidth and planar graphs, we propose, when dealing with a planar embedded graph, to restrict our attention to tree decompositions that respect the embedding of the graph.  Recall that a tree decomposition is a tree whose nodes are mapped to subsets of vertices, called \emph{bags}, that satisfy three properties that we state in the sequel. We say that a tree decomposition respects an embedding if, for every bounded face $f$, at least one bag contains every vertex of $f$.  We call the minimum width of such a decomposition the embedded-width or em-width of the embedding.

The purpose of this proposed width measure is two-fold.  
First, we note that algorithms for planar graphs frequently make explicit use of the embedding in their execution (for example, considering edges in clockwise order).  Further, algorithms for planar graphs that rely on tree decompositions often exploit the embedding of the graph in order to build the decomposition.  In fact, plane graphs whose dual graphs have bounded depth have natural tree decompositions that respect the embedding of the primal graph (though they may not be minimum width)\cite{Eppstein95}.
We believe that decompositions that respect the embedding will prove more useful than standard tree decompositions in algorithm design.

Second, we hope that by better understanding tree decompositions that respect the embedding of plane graphs we may be able to make progress toward resolving the long-open problem of determining the complexity of computing the treewidth of planar graphs by giving a tool for exploiting the embeddings of planar graphs in calculating decompositions.  In fact, because every tree decomposition of a planar triangulation respects the embedding and every plane graph can be triangulated without increasing the treewidth (Biedl and Vel\'{a}zquez \cite{TL2013}) we immediately have the following:

\begin{theorem}
If a planar graph $G$ has treewidth at least 3, then there exists a triangulation $\hat{G}$ of an embedding of $G$ such that the em-width of $\hat G$ is equal to the treewidth of $G$.
\end{theorem}

We discuss further connections between treewidth and em-width in Section~\ref{sec:emWidthBounds} following formal definitions.  In Section~\ref{sec:algorithm}, we adapt an algorithm of Bodlaender \cite{Bodlaender93,BDDFLP16} for computing treewidth to compute em-width.  This recursive algorithm requires the existence of a large matching for the recursion to sufficiently reduce the size of the graph; we give such bounds in Section~\ref{sec:match}.

\section{Definitions}

We use standard graph theoretic notation.  $G[X]$ is the subgraph of $G$ induced by a set of vertices $X$.

An \emph{embedding} of a planar graph is a mapping from vertices to distinct points in the plane, and from edges to non-intersecting curves (whose endpoints are the images of the vertices the edge connects).  A planar graph with a given embedding is called a {\em plane graph}.
A \emph{face} of a plane graph $G$ is a component of the complement of the image of $G$. The \emph{boundary} of a face $f$, denoted $\partial f$, is the subgraph of $G$ whose image is contained in the closure of $f$. $V(\partial f)$ denotes the set of vertices of $\partial f$.  The \emph{length} of a face is the number of vertices on its boundary. We write $f_{\infty}$ to denote the unique unbounded face of $G$.

The \emph{dual} of a plane graph $G$ is another plane graph $G^*$: $G^*$ has a vertex for every face of $G$ and for every edge $e$ of $G$, there is a corresponding edge $e^*$ whose endpoints are the faces which have $e$ on their boundary.  For a vertex $u$ of $G$, we denote the corresponding face in $G^*$ by $u^*$.

An embedding is \emph{outerplanar} if all of the vertices are contained in $\partial f_{\infty}$. More generally, a graph is $k$-outerplanar if removing all vertices in $\partial f_{\infty}$ results in an $(k-1)$-outerplanar graph, where a 1-outerplanar graph is an outerplanar graph.  We make use of the following equivalent definition: Assign the label $1$ to every vertex on the outer face; label each face with the smallest label that appears on a vertex on their boundary; and label each other vertex with one more than the minimum label of its incident faces. Then an embedding is $k$-outerplanar if $k$ is the largest label on a vertex.  We will call this labeling the \emph{outerplanarity labeling}.

\subsection{Treewidth} \label{sec:twDef}

\begin{definition} A \emph{tree decomposition} of a graph $G=(V,E)$ is a tuple $(T = (I,F), \{X_i \mid i \in I\})$ where $T$ is a tree and $X_i \subseteq V$ for each index $i\in I$, such that
\begin{enumerate}
	\item Every vertex lies in a bag: $\bigcup_{i\in I} X_i = V$,
	\item Every edge lies in a bag: For every edge $(u,v)$, there is a bag $X_i$ that contains both $u$ and $v$, and
	\item Vertices induce subtrees: For every vertex $v$, the index set $\{i\in I \mid v \in X_i\}$ induces a connected subgraph of $T$.
\end{enumerate}
	The vertex subsets $X_i$ are called \emph{bags}. The \emph{width} of a decomposition is one less than the size of the largest bag, and the \emph{treewidth} of a graph, denoted $\tw(G)$, is the minimum width across all tree decompositions of $G$.
\end{definition}
	
\noindent We use several well-known observations about treewidth:

\begin{enumerate}
\item  \label{obs:twdual}
$\tw(G^*) \leq \tw(G) + 1$. \cite{Lapoire96,Bouchitte01}
\item \label{obs:grid}
  The treewidth of the $p \times q$ grid is $\min\{p,q\}$. \cite{Bodlaender94}
\item \label{obs:subdividing} Subdividing an edge does not change the treewidth of a graph. \cite{Lozin06}
\item \label{obs:cliqueBag} If $S \subseteq V$ form a clique in $G$ then every tree decomposition of $G$ contains a bag $X_i$ such that $S \subseteq X_i$. \cite{Bodlaender88Dynamic}
\item \label{obs:tw3k1} The treewidth of a $k$-outerplanar graph is at most $3k-1$.  \cite{Bodlaender88}
\end{enumerate}

\subsection{Embedded-width}

When we restrict our attention to planar graphs, we often wish to make use of a particular embedding of that graph.  However, tree~decompositions are embedding-agnostic.  We define a special case of tree decompositions that incorporates the embedding of a plane graph by requiring that the vertices of each bounded \emph{face} lie in a bag.  Formally:

\begin{definition}
An \emph{embedded tree decomposition} (or {\em em-decomposition}) of a plane graph $G = (V,E)$ is a tree decomposition $(T,{\cal X})$ of $G$ that additionally satisfies a fourth property:
\begin{enumerate}\setcounter{enumi}{3}
	\item for each bounded face $f$ of $G$, there is a bag $X\in \mathcal{X}$ that contains every vertex in $V(\partial f)$.
\end{enumerate}
The \emph{width} of an em-decomposition is the size of the largest bag minus one. The \emph{embedded-width} or \emph{em-width} of a plane graph, denoted $\emw(G)$, is the minimum width across all em-decompositions of $G$.
\end{definition}

We restrict ourselves to the bounded faces of a plane graph so that many natural widths still hold: the em-width of a tree is 1 and the em-width of every outerplanar triangulation is 2.  This restriction also implies that the tree decompositions of a triangulated plane graph whose dual has bounded depth are embedded tree decompositions~\cite{Eppstein95}.

We make use of the following to connect the em-decomposition and tree decomposition of a plane graph. Define the \emph{facial completion} $\widetilde{G}$ of a plane graph $G$ by adding to each bounded face $f$ a clique on the vertices of $\partial f$.

\begin{lemma} \label{lem:completionWorks}
A decomposition is an em-decomposition of $G$ if and only if it is a tree decomposition of $\widetilde{G}$.
\end{lemma}
\begin{proof} 
An em-decomposition of $G$ satisfies properties~1 and~3 of being a tree decomposition for $\widetilde{G}$ by the same properties of em-decompositions.  Every edge $uv$ of $\widetilde{G}$ that is not an edge of $G$ connects two vertices on the same bounded face of $G$.  By property~4 of em-decompositions, there is a bag containing both $u$ and $v$, so the em-decomposition of $G$ satisfies property~2 of being a tree decomposition for $\widetilde{G}$.

A tree decomposition of $\widetilde{G}$ is also a tree decomposition of $G$ (because $G$ is a subgraph of $\widetilde{G}$). Property~4 follows from Observation~\ref{obs:cliqueBag}.
\qed \end{proof}

\section{Bounds on embedded-width} \label{sec:emWidthBounds}

We give upper bounds on the em-width of a plane graph in terms of its treewidth and the length of its longest face.   Specifically, we prove, in Section~\ref{sec:twbounds}:
\begin{theorem} \label{thm:emwtwAll}
If $G$ is plane graph where every bounded face has length at most $\ell$ then \[\emw(G) \leq (\tw(G) + 2)\cdot \ell - 1.\]
\end{theorem}
In combination with Bodlaender's bound on the treewidth of $k$-outerplanar graphs (Observation~\ref{obs:tw3k1}), Theorem~\ref{thm:emwtwAll} implies that the em-width of a $k$-outerplanar graph where every bounded face has length at most $\ell$ is at most $(3k+1)\ell-1$.  We obtain a slightly tighter upper bound using a more direct proof:

\begin{theorem} \label{thm:emw3kl}
  If $G$ is a $k$-outerplanar graph where every bounded face has length at most $\ell$, then \[\emw(G) \leq 3k\ell - 1.\]
\end{theorem} 

These bounds are optimal up to constant factors:
\begin{lemma} \label{lem:lbemwtwAll}
For all $t$ and all sufficiently large $n$,
there is a plane graph $G$ with $n$ vertices with treewidth $t$, where all bounded faces have length at most $\ell$, such that
\[\emw(G) \ge (\ell/2-1)(t-1).\]
\end{lemma}

\begin{lemma} \label{lem:emwoplb}
For all $k$ and all sufficiently large $n$, there is a $k$-outerplanar graph $G$ with $n$ vertices, where all bounded faces have length at most $\ell$, such that
\[\emw(G) \ge (\ell/2-1)(2k-1).\] 
\end{lemma}

\begin{proof}[of Lemmas~\ref{lem:lbemwtwAll} and~\ref{lem:emwoplb}]
Let $K$ be a $p\times q $ grid with $q > k(p-1)$.  Let $G$ be the plane graph obtained from $K$ by subdividing each vertical edge $k-1$ times.  Then each bounded face of $G$ has length $\ell = 2k+2$. If $G$ has fewer vertices than the required $n$, add a path of the required length to the outer face of $G$.

Let $\widetilde{G}$ be the facial completion of $G$. By Lemma~\ref{lem:completionWorks}, $\tw(\widetilde{G}) = \emw(G)$. Since $\widetilde{G}$ contains a $k(p-1)\times k(p-1)$ grid as a subgraph, $\tw(\widetilde{G}) \geq k(p-1)$ (Observation~\ref{obs:grid}) and therefore $\emw(G) \ge k(p-1)$.  Since each bounded face of $G$ has length $\ell = 2k+2$, we have $\emw(G) \ge (\ell/2-1)(p-1)$.

By Observation~\ref{obs:grid}, $\tw(K) = p$.   By Observation~\ref{obs:subdividing}, $\tw(G) = \tw(K) = p$. This completes the lower bound in terms of treewidth (Lemma~\ref{lem:lbemwtwAll}).

By additionally observing that $G$ has outerplanarity at least $\lceil p/2 \rceil$, we get the lower bound in terms of outerplanarity (Lemma~\ref{lem:emwoplb}).
\qed \end{proof}

\subsection{Weak duals}

We use the weak dual of a given graph to prove our bound. The weak dual, $G^+$, of a planar graph, $G$, is obtained from $G^*$ by deleting $f_\infty^*$~\cite{FGH74}; namely, $G^+ = G^* \setminus \{f_\infty^*\}$.  Our general approach for proving the bounds of Theorems~\ref{thm:emwtwAll} and~\ref{thm:emw3kl} is to find a tree decomposition of $G^+$ and then convert it to an em-decomposition of $G$.  Note that the weak dual of a graph is a subgraph of its standard dual, thus $\tw(G^+) \leq \tw(G^*)$. Combining this with Observation~\ref{obs:twdual}, we get:
\begin{equation} \label{eqn:twweakdual}
\tw(G^+) \leq \tw(G) + 1.
\end{equation}

For a vertex $u$ of $G$,  we use $u^+$ to denote the set of all vertices of $G^+$ that correspond to faces of $G$ with $u$ on their boundary; that is, $u^+ = V(\partial u^*) \setminus \{ f_\infty^* \}$.

\begin{lemma} \label{newdualConn} 
	If $G^+$ is connected then for every vertex $u$ of $G$, $G^+[u^+]$ is connected. Further, for a vertex $u$ of $G$, $G^+[u^+]$ is disconnected if and only if $u$ is a cut vertex of $G$ with at least 4 incident edges in $\partial f_\infty$.
\end{lemma}
\begin{proof} 
  Consider the edges $W$ incident to $u$ in clockwise order of their embedding around $u$: $e_1, e_2, \ldots, e_{d(u)}$. By definition of the dual, $e_1^*, e_2^*, \ldots, e_{d(u)}^*$ forms a closed walk $W^*$ in $G^*$.   $W^*$, although not necessarily simple, is such that every edge of $W^*$ bounds the face $u^*$.  Since $W^* = \partial u^*$, $G^*[u^*]$ is connected.  (Note that the same dual vertex may appear multiple times in $W^*$.)   

We prove the contrapositive of the first statement; we assume $G^+[u^+]$ is disconnected.
Note that $W^*\setminus f_\infty \subseteq G^+[u^+]$.  Therefore, 
$G^+[u^+]$ is disconnected if and only if $f_\infty$ appears at least twice on $W^*$ as endpoints of non-self-loops (in order to create two components of $W^*$); this proves the second part of the second statement of the lemma.  Let $W_1^*$ and $W_2^*$ be edge-disjoint  $f_\infty$-to-$f_\infty$ subpaths of $W^*$ that partition the edges of $W^*$.  In the plane, $W_1^*$ and $W_2^*$ are closed curves that together partition the faces of $G^*$ (and so also the vertices of $G$), into (at least) three sets: $\{u^*\}$, $A$ and $B$.

Consider an $\alpha$-to-$\beta$ path $P^*$ in $G^*$ for $\alpha \in A$ and $\beta \in B$.  Since every edge in $W^*$ bounds $u^*$, to go from a region bounded by $W_1^*$ to a region bounded by $W_2^*$, $P^*$ must visit $f_\infty$.  Therefore, removing $f_\infty$ disconnects $G^*$: $G^+$ is disconnected.

Likewise, consider an $a$-to-$b$ path $P$ in $G$ where $a^* = \alpha$ and $b^* = \beta$.  As curves in the plane, $P$ must cross $W_1^*$ and $W_2^*$.  However, the edges of $W_1$ and $W_2$ are all incident to $u$.  Therefore, removing $u$ disconnects $G$: $u$ is also a cut vertex, proving the second part of the second statement of the lemma.
\qed \end{proof}

\subsection{Relating em-width and treewidth} \label{sec:twbounds}

We first prove Theorem~\ref{thm:emwtwAll} for planar graphs with connected weak duals. 
\begin{lemma} \label{lma:connWeakDual}
If $G$ is plane graph with all faces bounded by at most $\ell$ vertices and with a connected weak dual, then \[\emw(G) \leq (\tw(G) + 2)\cdot \ell - 1.\]
\end{lemma}
\begin{proof}
First observe that since $G$ has a connected weak dual, by Lemma~\ref{newdualConn}, the subgraph $\partial f_\infty$ can be partitioned into a cycle $C$ and a set of trees $F_1, F_2, \ldots$ where tree $F_i$ has a single vertex $v_i$ in common with $C$.

We prove the lemma constructively by converting a minimum-width tree decomposition $(T^+, \mathcal{X}^+)$ of $G^+$ into an em-decomposition $(T, \mathcal{X})$ of $G$ in two steps:
\begin{enumerate}
\item Set $T = T^+$ and for each $X^+ \in \mathcal{X}^+$, create a bag $X \in \mathcal{X}$ such that $X$ contains all vertices of $G$ that are on the boundary of faces corresponding to vertices of $X^+$; namely $X = \{u\ |\ u \in \partial \alpha\ s.t.\ \alpha^* \in X^+\}$.
\item The vertices of $G$ that are not yet represented by $(T, \mathcal{X})$ are those vertices that are {\em only} in $\partial f_\infty$ (that is, not on the boundary of any finite face of $G$).  These are exactly the vertices of $F_1, F_2, \ldots$ not in $C$.  For each $i$, we add a 
minimum-width tree decomposition $(T_i,\mathcal{X}_i)$ of $F_i$ to $(T,\mathcal{X})$ by connecting a node of $T_i$ to a node of $T$ where both bags contain $v_i$.
\end{enumerate}

\begin{claim}
  $(T,\mathcal{X})$ is an em-decomposition of $G$.
\end{claim}
\begin{proof} 
  Properties~1 and~2 of em-decompositions hold since every vertex and edge of $G$ is either on the boundary of a finite face (and so included in a bag in step~1) or not (and so included in a bag in step~2.)

  Property~4 of em-decompositions holds since every non-outer face $\alpha$ is in some bag of $\mathcal{X}^+$, and so all the vertices of $\partial \alpha$ are included in a bag of $\mathcal{X}$ in step~1.
  
 All that remains is to illustrate the connectivity requirement (property~3).  Consider a vertex $u$ added to a bag in step~1.  Since $u$ is on the boundary of a finite face, $u$ will be added to the bags of $\mathcal{X}$ that correspond to bags of $\mathcal{X}^+$ that contain finite primal faces for which $u$ is on the boundary.  Specifically, $u$ is in the set of bags $\{X \in {\cal X} \ : \ u \in \partial \alpha\ s.t.\ \alpha^* \in X^+\} = \{X \in {\cal X}\ : \ \alpha^* \in X^+\ s.t.\ \alpha^* \in G^+[u^+] \} $.  By Lemma~\ref{newdualConn}, since $G^+$ is connected, $G^+[u^+]$ is connected; since $(T^+, \mathcal{X}^+)$ satisfies the connectivity requirement, the bags of ${\cal X}^+$ that contain vertices of $G^+[u^+]$ are also connected in $T^+$.  Now consider the vertices added to bags in step~2. The bags of $(T_i, {\cal X}_i)$ satisfy property~3 and only vertices connecting the trees of $\partial f_\infty$ to the rest of the graph (e.g.\ $v_i \in F_i$) appear in both bags added in step~1 and step~2.  However, by construction, these bags are connected in the last part of the construction of step~2.
\qed \end{proof}
\begin{claim} 
  The width of $(T,\mathcal{X})$ is at most $(\tw(G) + 2)\cdot \ell - 1$.
\end{claim}
\begin{proof}
  Each bag of $(T^+, \mathcal{X}^+)$ has size at most $\tw(G^+) + 1$ which is $ \leq tw(G)+2$ by Equation~(\ref{eqn:twweakdual}).  The bags of $(T,\mathcal{X})$ derived from $(T^+, \mathcal{X}^+)$ are bigger by a factor of at most $\ell$, the bound on the size of the faces of $G$, and so have size $ \leq (tw(G)+2)\cdot \ell$.  The bags added to $(T,\mathcal{X})$ that correspond to the trees $F_1, F_2, \ldots$ in $\partial f_\infty$ each have size at most 2 since trees have tree-width 1.
\qed \end{proof}
\noindent This completes the proof of Lemma~\ref{lma:connWeakDual}.
\qed \end{proof}

We can now generalize this result to all plane graphs (eliminating the condition that weak dual must be connected).

\begin{proof}[of Theorem~\ref{thm:emwtwAll}]
Consider a ``pseudo'' block decomposition of $G$.  
A block decomposition decomposes the edges of $G$ into biconnected components that are connected to each other in a tree structure via cut vertices of $G$.  We, rather, use a coarser decomposition $\cal B$ cutting $G$ at only the cut vertices incident to $f_{\infty}$ more than once. That is, each $B \in {\cal B}$ is a subgraph of $G$, all boundary edges of which are on $f_\infty$, two components of ${\cal B}$ have a single cut vertex in common and the components of $\cal B$ are connected in a tree structure.

By construction, each $B \in {\cal B}$ is either an edge, or is a graph whose outer face has a cycle for a boundary, so by Lemma \ref{newdualConn}, these subgraphs have connected weak duals. 
 By Lemma~\ref{lma:connWeakDual}, $B$ has an em-decomposition $(T_B, {\cal X}_B)$ of width at most $(\tw(G) + 2)\cdot \ell - 1$.  

We connect these em-decompositions as follows.  For two subgraphs $B_1$ and $B_2$ of ${\cal B}$ that share a cut vertex $s$, connect a bag of $(T_{B_1}, {\cal X}_{B_1})$ that contains vertex $s$ to a bag of $(T_{B_2}, {\cal X}_{B_2})$ that contains vertex $s$.  This maintains properties~1,~2, and~4 of em-decompositions, since  $(T_{B_1}, {\cal X}_{B_1})$ and $(T_{B_2}, {\cal X}_{B_2})$ are valid em-decompositions.  Property~3 is achieved for $B_1 \cup B_2$ since bags containing $s$ are connected.  The width of the union of the decompositions is not increased.  Repeating for each pair of subgraphs of $\cal B$ that share a vertex results in a valid em-decomposition since the subgraphs are connected in a tree structure.
\qed \end{proof}

We can likewise relate the em-width to the outerplanarity of a plane graph by way of weak duals. The proof technique is similar, so we defer the proof of Theorem \ref{thm:emw3kl} to the Appendix.

\section{Matchings in planar graphs}\label{sec:match}
In Section~\ref{sec:algorithm}, we discuss an algorithm for computing em-width. The algorithm recurses on smaller graphs created by contracting a matching in the current instance. To be efficient, our algorithm requires that each instance have a large matching. 
Not all planar graphs have matchings of size $\Omega(n)$; $K_{2,r}$ is a counterexample for large $r$. We argue that these large $K_{2,r}$ are actually the only obstructions to large matchings in planar graphs with minimum degree $2$, which will be sufficient for the purposes of our algorithm. 

Call a set of $r$ degree-$2$ vertices with the same set of neighbors an ``$r$-family.'' Note that an $r$-family and its common neighbors induce a $K_{2,r}$ subgraph, possibly with an additional edge between the vertices in the size-$2$ partite set.

\begin{theorem} \label{noK2r}
If $G$ is a connected planar graph with minimum degree 2 and no $r$-family ($r \geq 3)$, then $G$ has a matching of size at least $\frac{n}{12r-3}$. 
\end{theorem}

Our approach is as follows. A theorem of Nishizeki \cite{Nishizeki79} says that planar graphs of minimum degree $3$ have large matchings. To extend to degree $2$ graphs with no $K_{2,r}$ we divide into cases. Let $G$ be our input planar graph. If there are few degree $2$ vertices, we can find a supergraph $H$, of $G$ which has minimum degree $3$, and such that every large matching of $H$ contains a large matching of $G$. Otherwise, there are many degree $2$ vertices. Either, many of them are adjacent to each other and those paths and cycles have large matchings, or many are adjacent to only high-degree vertices and the lack of $K_{2,r}$ subgraphs guarantees many high degree vertices as well, from which we find a matching. The full details of the proof are deferred to the Appendix.

Setting $r = 3$ in Theorem~\ref{noK2r}, we get the following special case:
\begin{corollary} \label{cor:noK23}
If $G$ is a connected planar graph with minimum degree 2 and no $3$-family, then $G$ has a matching of size at least $\frac{n}{33}$.
\end{corollary}

We strengthen this result (and get a slightly less tight bound) by allowing $r$-families along as they are not {\em nicely embedded} in a given plane graph.  Recall that an $r$-family refers to the vertices of degree 2 (that have common neighbors).

\begin{definition}[Nicely embedded $r$-family]
  Let $H$ be the subgraph of a plane graph given by the edges incident to the vertices of an $r$-family $R$.  If for every finite face $f$ of $H$, $\partial f$ (as a cycle in $G$) does not strictly enclose any vertices of $G$, then $R$ is nicely embedded.  (See Figure \ref{rFamiliesPic}.)  If $\partial f$ strictly encloses a set of vertices $O$ of $G$, we call the subgraph $G[O]$ an obstruction to the nice embedding of $R$. 
\end{definition}

\begin{figure}[h] 
\includegraphics[width=\textwidth]{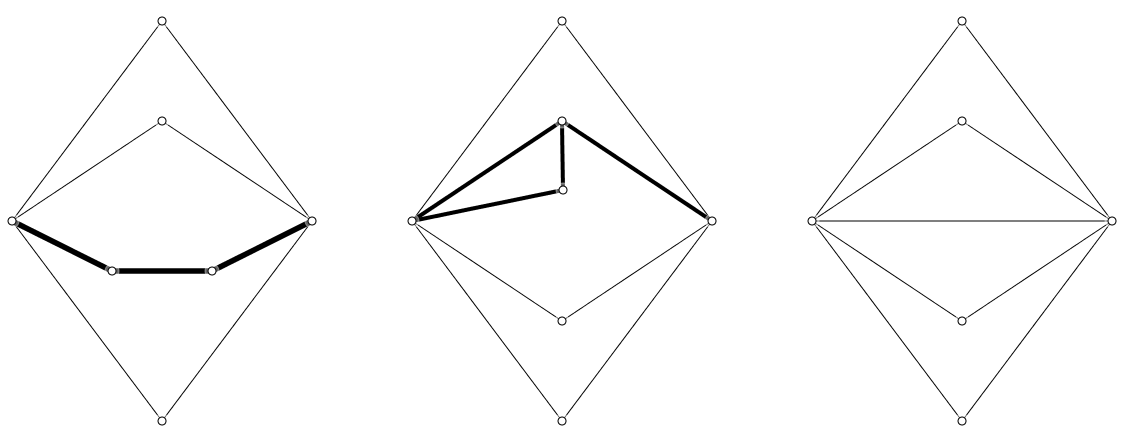}
\caption{Two $3$-families (left, center) which are not nicely embedded and a nicely embedded $4$-family (right). The obstructions to the nice embeddings of the 3-families are given by bold edges. Note that an edge between the common neighbors does not violate being nicely embedded.}
\label{rFamiliesPic}
\end{figure} 

\begin{theorem} \label{noNiceK23}
If $G$ is a connected plane graph of minimum degree $2$ with no nicely embedded $r$-family for $r \ge 3$, then $G$ has a matching of size at least $\frac{n}{37}$.
\end{theorem}
\begin{proof}
We divide into two cases based on the number of vertices which belong to (maximal) $r$-families.  By assumption, none of these $r$-families are nicely embedded.  Let there be $p$ such vertices.

\noindent {\bf Case 1:} $\mathbf{p \geq \frac{4n}{37}}$\\ 
Consider the subgraph $H$ given by the edges incident to a maximal $r$-family of $G$.  Since no $3$-family of $G$ can be nicely embedded, for every two finite faces $f,g$ of $H$ that share edges, there must be an obstruction in $G$ enclosed by at least one of $\partial f,\partial g$ (since the edges of $\partial f$ and $\partial g$ witness a 3-family that also must not be nicely embedded.  That is, there are obstructions in $G$ embedded in at least half the finite faces of $H$.  There must be at least $\lceil \frac{r-1}{2} \rceil$ obstruction for this $r$ family.  Since $\lceil \frac{r-1}{2} \rceil \ge \frac{r}{4}$, in $G$ there are at least $\frac{p}{4}$ obstructions.

Since $H$ corresponds to a maximal $r$-family, and $G$ has minimum degree 2, each obstruction must contain an edge.  Since the obstructions are vertex disjoint, taking one edge from each obstructions gives a matching of size at least $\frac{p}{4} \geq \frac{n}{37}$.

\noindent {\bf Case 2:} $\mathbf{p \leq \frac{4n}{37}}$\\
For each maximal $r$-family ($r \geq 3$), delete all but two of the degree-$2$ vertices. Call this new graph $G'$. $G'$ is a $3$-family-free graph with at least $n-p$ vertices, so by Corollary~\ref{cor:noK23}, there is a matching of size $\frac{n-p}{33} \geq \frac{n}{37}$.
\qed \end{proof}

\section{Calculating em-width} \label{sec:algorithm}

In this section we develop a fixed-parameter tractable algorithm for deciding whether a planar graph has em-width at most $k$ (and constructing the embedded tree decomposition if it exists). 
We adapt an algorithm of Bodlaender for deciding whether a graph has  treewidth at most $k$ (and constructing the tree decomposition if it exists)~\cite{Bodlaender93}.  Our algorithm, like Bodlaendar's, runs in linear time in the size of the graph (ignoring the run-time dependence on $k$).

One option is to simply do the following: for input plane graph $G$, construct the facial completion $\widetilde{G}$ and apply Bodlaender's treewidth algorithm to $\widetilde{G}$.  The result (by Lemma~\ref{lem:completionWorks}) is an em-decomposition for $G$. (Indeed, we take this approach for one of the subroutines of our algorithm.)

However, as one of the motivations for studying em-width is to better understand how embeddings affect decompositions, it is instructive to see what role the embedding can play in the explicit execution of the algorithm, rather than simply using a black box. Moreover, our algorithm (by exploiting planarity and the embedding of the graph) is simpler conceptually than Bodlaender's. We therefore include our version as evidence of our claim in the introduction that explicitly considering the embedding proves useful in algorithm design.

Further, our algorithm is faster than  Bodlaender's algorithm by a polynomial factor of $k$ as a result of a more efficient recursion. We defer the comparison to the Appendix.  

\subsection{Core Algorithm Description}
Like Bodlaender's algorithm for treewidth, we will contract along matchings to decrease the graph size. Following Theorem \ref{noNiceK23}, we will find a subgraph of our current instance that has no nicely embedded $3$-family and minimum degree $2$. This gives us a large matching, which we contract. When we contract along this matching we may create multi-edges. We cannot delete these repeated edges automatically, as they could form the border of a face. However if the multi-edges have nothing embedded between them, one copy can be deleted without affecting the faces or the em-width. After deleting any multi-edges of this type, we recurse and find an em-decomposition of this subgraph, this decomposition can be easily converted into a (non-optimal) em-decomposition for the current instance. Finally we improves the non-optimal em-decomposition into one of the desired width. A more precise description follows.
 
\noindent Given a plane graph $G$:
\begin{enumerate}
	\item Remove degree-$1$ vertices.
	\item Identify all nicely embedded maximal $r$-families ($r \geq 3)$.  For each $r$-family $F$, order the vertices of $F$ as $f_1, f_2, \ldots, f_r$ such that $f_i$ and $f_{i+1}$ appear on the boundary of a common face.  Delete $f_2, \ldots, f_{r-1}$ from the graph.
	\item  Find a maximal matching and contract the edges of the matching.
	\item  \label{GPrimeFinished} For every pair of parallel edges that form the boundary of a face, remove one edge.
	\item Recurse on the resulting graph, $G'$. Let $T$ be the em-decomposition returned by the recursive call.
	\item For every vertex $u$ in $G'$ which was a contraction of $(v,w)$ in $G$, replace appearances of $u$ in $T$ with $v$ and $w$.
	\item For every $r$-family $F$, there must be a node $x$ in $T$ whose bag contains $f_1,f_r$ and the common neighbors of $f_1$ and $f_r$, $a$ and $b$.  Attach to $x$ a path of nodes whose bags are\\ $\{a,b,f_1,f_2,f_r\},\{a,b,f_2,f_3,f_r\},\{a,b,f_3,f_4,f_r\},\ldots,\{a,b,f_{r-2},f_{r-1},f_r\}$.
	\item For each deleted degree-$1$ vertex, $u$, if $u$ is incident to some interior face, create a bag with only the boundary of that face, and attach it to the bag that contained the boundary of the corresponding face in $G'$. If $u$ was not incident to an interior face, create a bag with it and its neighbor and attach it to any bag containing the neighbor.
	\item Run the decomposition improvement algorithm (see Section \ref{improveAlg}) on $T$ to get a decomposition of width $k$ (if one exists).
\end{enumerate}

\begin{lemma}
$\emw(G') \leq \emw(G)$.
\end{lemma}
\begin{proof}
It is enough to show that we can convert any em-decomposition of $G$ into an em-decomposition of $G'$ without increasing the width. We can accomplish this conversion as follows: delete appearances of the degree-$1$ vertices from the decomposition. For contracted edges $(u,v)$ replace appearances of $u$ and $v$ with the combined vertex in $G'$. 

Finally, for an $r$-family $f_1, f_2, \ldots, f_r$ with common neighbors $a$ and $b$, replace appearances in bags of the em-decomposition of $r_i$ for $r = 2, \ldots, r-1$ with $r_1$.  Since $r_i$ and $r_{i+1}$ appear in some common bag (since they bound a common face), property 3 of em-decompositions is met.  Since $a,f_{r-1},b,f_r$ formed a face in $G$, $\{a,f_{r-1},b,f_r\}$ will be in a common bag of the new decompositions, meeting property 4 of em-decompositions.

Thus this is a valid em-decomposition for $G'$. Since, we have not increased the size of any bag, the width did not increase.
\qed \end{proof}
Therefore, when $T$ is returned from the recursive call it has bags of size at most $k+1$ (where $k$ is the em-width of $G$). Expanding the matching could double the bag size to $2k+2$. All of the remaining bags that are added are boundaries of faces, with at most one additional vertex added, thus are size at most $\ell + 1$.  Thus the width of the decomposition $T$ is at most $\max\{\ell+1, 2k+2\}-1$ which is at most $2k+1$ since $\ell \leq k + 1$.

We show how to improve the width of this decomposition (if it exists) from $2k+1$ to $k$ in Section \ref{improveAlg}.

\subsection{Run-time analysis}

To properly discuss running time we will require the following lemma:
\begin{lemma} \label{edgesOn}
	The number of edges of $G'$ (counting multi-edges) is linear in the number of vertices of $G'$.
\end{lemma}
\begin{proof}
The underlying simple graph is planar, so it is enough to show that the number of repeated edges is linear in the number of vertices. If there is more than one copy of some edge $(u,v)$, then for each consecutive pair of edges $e$ and $e'$, there must be a vertex, $w$, incident to $u$ or $v$ and embedded inside the Jordan curve formed by $e$ and $e'$ (i.e.\ an obstruction to a nicely-embedded family). We can create an injective mapping from the additional copies of $(u,v)$ to these $w$ vertices. Since the mapping is injective, there can only be at most a linear number of multi-edges, so the total number of edges of $G'$ is linear in the number of vertices as claimed.
\qed \end{proof}

By Lemma \ref{edgesOn}, we can therefore use ``linear'' to mean linear in the number of vertices, or equivalently the number of edges in the (multi-) graph $G'$.
 
We can find all $r$-families in linear time as follows: sort the set of degree-2 vertices by their neighbors (for example, the sum of the labels of their neighbors).  Now $r$-families are consecutive in this order and easily identifiable.  Identifying nicely-embedded $r$-families from this can also be done in linear time: for each neighbor of an $r$-family, iterate through its incident edges in their cyclic order in the embedding. Assign labels to edges incident to degree-$2$ vertices, using the same label until an edge not incident to a degree-$2$ vertex is found (at which point we switch labels). We can then radix sort degree-$2$ vertices again (by these assigned labels). Every nicely embedded family has the same labels on its incident edges (and so appear consecutively in the sort).

To find parallel edges bounding a common face: for every multi-edge $(u,v)$, choose one of the edge copies $e$. Find the next edge in clockwise order around $u$ and in counter-clockwise order around $v$. If these are the same edge, then delete $e$. Iterating this process for all $e=(u,v)$ and for all adjacent vertices $u$ and $v$ removes all consecutively embedded edges. Since there are a constant number of lookups per edge, this takes linear time.

All other steps besides the decomposition improvement algorithm can easily be done in linear time as well.

In Section \ref{improveAlg}, we show that the decomposition improvement algorithm runs in time linear in the size of $G$ (though exponential in the desired em-width). By Theorem \ref{noNiceK23}, an $\Omega(n)$ matching always exists, and the maximal matching we find is a $2$-approximation of the optimal, so the problem size decreases by a constant factor at each step. Thus the overall running time is linear in the size of the graph (and exponential in the desired em-width).

\subsection{Improving Good Decompositions} \label{improveAlg}
We use the following algorithm of Bodlaender and Kloks as a black box:

\begin{theorem}[Bodlaender and Kloks \cite{Bodlaender91}] \label{BodlaenderImprove}
For all constants $h,k \in \mathbb{Z^+}$ there is an algorithm which, given a graph $G$ and tree decomposition for $G$ of width $h$, produces a tree decomposition for $G$ of width $k$ or determines that such a decomposition does not exist, in $O(n)$ time.
\end{theorem}

Bodlaender and Kloks describe these algorithms explicitly (using a dynamic programming approach). 
Let $T$ be an em-decomposition  of $G$ of width at most $2k+1$ as guaranteed by the first 8 steps of our algorithm. Construct the facial completion  $\widetilde{G}$ of $G$, and input $T$ and $\widetilde{G}$ to the algorithm of Theorem~\ref{BodlaenderImprove}. This returns a tree decomposition, $T'$, of $\widetilde{G}$ of width $k$ (or indicates that none exists). By Lemma~\ref{lem:completionWorks}, if $T'$ exists, it is also an em-decomposition of $G$ of width $k$.

We can construct $\widetilde{G}$ in $O(k^2 n)$ time: by Euler's Formula, there are $O(n)$ faces, and each interior face is of length at most $\ell \leq k$ so we add at most $O(k^2)$ edges (requiring at most $k^2$ time) to add a clique to each face.

This completes the analysis of our algorithm to calculate em-width. 

\paragraph{Acknowledgements} This work was completed while Robbie Weber was participating in a Research Experience for Undergraduates program at Oregon State University and while an undergraduate at the University of Illinois at Urbana-Champaign.  This material is based upon work supported by the National Science Foundation under Grant Nos.\ CCF-1252833 and CCF-1408763.

\bibliographystyle{plain}
\bibliography{width_bib}

\appendix

\section{Proof of Theorem \ref{thm:emw3kl}}
We require the following:
\begin{lemma} \label{lem:dualOP}
	If $G$ is a $k$-outerplanar graph then $G^+$ is at most $k$-outerplanar.
\end{lemma}
\begin{proof} 

Consider the vertex labeling $\ell\colon V \rightarrow [1,k]$ such that for all $j$, the vertices labeled $j$ are on the boundary of the graph induced by the vertices labeled $1$ through $j$.  (This is the reverse of the standard {\em outerplanarity labeling}.)  We derive an outerplanarity labeling $\ell^+$ of $G^+$ as follows.  Let $F_j$ be the set of finite faces of the subgraph of $G$ induced by the vertices labeled $1$ through $j$.  Set $\ell^+(\alpha^*) = j$ if $\alpha$ is a face in $F_j \setminus F_{j-1}$.  If $\alpha \in F_j \setminus F_{j-1}$, then $\partial \alpha$ contains a vertex on the boundary of the graph induced by the vertices labeled $1$ through $j$. Therefore, $\alpha^*$ is on the boundary of the dual graph $G^*$ induced by the vertices in $F_j$.  Since $\ell^*$ labels all finite faces of $G$, $\ell^+$ labels all vertices of $G^+$, and this labeling is a valid $k$-outerplanarity labeling of $G^+$.
\qed \end{proof}

\begin{proof}[of Theorem~\ref{thm:emw3kl}]
We use the same construction we used in the proof of Theorem~\ref{thm:emwtwAll}. Recall that we divide $G$ into a pseudo block decomposition, $\mathcal{B}$, where the blocks are nontrivial subgraphs with connected weak duals or single edges. For each nontrivial $B \in \mathcal{B}$ we construct $(T_B, \mathcal{X_B})$ by finding $(T_B^+, \mathcal{X}_B^+)$, an optimal tree decomposition of $B^+$, and converting each appearance of $\alpha$ to $\{u \in G | u \in \partial \alpha\}$. Since $G$ is $k$-outerplanar, every $B$ is at most $k$-outerplanar, so by Lemma~\ref{lem:dualOP} every $B^+$ is at most $k$-outerplanar. By Observation~\ref{obs:tw3k1}, each bag of $(T_B^+, \mathcal{X}_B^+)$ is of size at most $3k$. By construction each bag of $(T_B, \mathcal{X}_B)$ contains at most $3k\ell$ vertices (where $\ell$ is the maximum number of vertices in the boundary of a finite face). Thus the width is at most $3k\ell - 1$. 
\qed \end{proof}

\section{Proof of Theorem \ref{noK2r}}

We begin with some needed results: 

\begin{theorem}[Robbins \cite{Robbins39}] \label{Robbins}
	An undirected graph, $G$, has a strongly-connected orientation if and only if $G$ is $2$-edge-connected.
\end{theorem}

\begin{lemma} \label{orient}
	If $G$ is a connected undirected graph, then $G$ has an orientation with at most one sink vertex, a vertex with outdegree 0.
\end{lemma}
\begin{proof}
  We orient every non-trivial maximal $2$-connected component of $G$ to have no sink (as guaranteed by Theorem~\ref{Robbins}).  To orient the remaining edges, we contract each of these components, which results in a tree; we pick a root $r$ of the tree and orient all the edges in the tree toward $r$.  If $r$ corresponds to a trivial $2$-connected component of $G$, $r$ will be a sink.  If $r$ corresponds to a non-trivial $2$-connected component of $G$, $G$ will have no sinks.
\qed \end{proof}

\begin{theorem}[Nishizeki \cite{Nishizeki79}] \label{deg3match}
If $G$ is a connected planar graph with minimum degree $3$ and at least $10$ vertices, then $G$ has a matching of size $\lceil \frac{n+2}{3} \rceil$.
\end{theorem}

We can now prove Theorem \ref{noK2r}
\begin{proof}
We may assume that $G$ has at least 33 vertices, for otherwise a single-edge matching satisfies the claim. Let $c$ be the fraction of vertices of $G$ which are degree $2$.

\noindent {\bf Case 1:} $\mathbf{c \leq \frac{2r-1}{4r-1}}$ \\
Fix a planar embedding of $G$. We construct a new graph $G'$ iteratively starting with $G' = G$.  While $G'$ has a degree-2 vertex $v$, add a vertex $v'$ and connect $v'$ to $v$ and $v$'s neighbors.  Since $v$ and $v$'s neighbors bound a common face, $v'$ can be embedded in this face, giving a planar embedding.  The resulting graph $G'$ is therefore planar. 
By construction, $G'$ has minimum degree $3$ and $n + s\cdot n$ vertices (where $s\cdot n \leq c\cdot n$ is the number of added vertices).  By Theorem \ref{deg3match}, $G'$ has a matching of size $\lceil \frac{1}{3}(n + sn + 2) \rceil$. At most $sn$ of these edges are incident to vertices not in $G$, so $G$ has a matching of size at least $\lceil \frac{1}{3}(n+sn + 2)\rceil - sn \geq \frac{1}{3}n - \frac{2}{3}sn \geq \frac{1}{3}n - \frac{2}{3} \cdot \frac{2r-1}{4r-1}n = \frac{n}{12r-3}$.

\noindent {\bf Case 2:} $\mathbf{c \geq \frac{2r-1}{4r-1}}$\\
Call a degree-2 vertex lonely if both its neighbors are degree $2$; otherwise call it social. Let $q$ be the fraction of degree-$2$ vertices which are social (i.e. there are $qcn$ social vertices and $(1-q)cn$ lonely vertices).  We have two subcases:

\noindent {\bf Case 2a:} $\mathbf{q \leq \frac{2r-2}{2r-1}}$\\
Let $H$ be the subgraph of $G$ of edges incident to lonely vertices.  $H$ is a graph of maximum degree 2 and so is a collection of paths or cycles.  Therefore a maximum matching $M$ of $H$ matches at least all but one vertex from each component of $H$.  Since each component in $H$ must have at least three vertices, $M$ matches at least two-thirds of the vertices of $H$.  So $|M| \ge \frac{1}{3}|V(H)| \ge \frac{1}{3}(1-q)cn\geq \frac{cn}{3(2r-1)} \geq \frac{(2r-1)n}{3(4r-1)(2r-1)} = \frac{n}{12r-3}$.

\noindent {\bf Case 2b:} $\mathbf{q \geq \frac{2r-2}{2r-1}}$\\
Let $S$ be the set of social vertices.  We construct a large matching $M$ of $G$ of edges incident to social vertices.  $G[S]$ is a set of $t_1/2$ edges and $t_2$ isolated vertices (where $t_1 + t_2 = qcn$) because at most one neighbor of a social vertex can have degree 2 (and so also be social).  We add the $t_1/2$ edges of $G[S]$ to $M$.  It remains to find a matching in $G$ of edges incident to (a subset) of the $t_2$ isolated social vertices.

Let $I$ be the set of $t_2$ isolated social vertices. Let $H$ be the subgraph of $G$ of edges incident to $I$. Let $F$ be the graph obtained from $H$ by contracting one edge incident to each vertex in $I$.  (Conversely $H$ is obtained from $F$ by subdividing every edge.) $F$ is a multigraph, but since $G$ has no $r$-family, $F$ has no $r$ edges in parallel.  Thus $F$ is a planar multi-graph with, by construction, exactly $t_2$ edges (one edge for each social vertex). Since $F$ has no $r$ edges in parallel, the underlying simple graph $\bar F$ has at least $\frac{t_2}{r-1}$ edges.  Since every simple planar graph has at most $3n-6$ edges, $\bar F$ (and $F$) has at least $\frac{t_2}{3(r-1)}$ vertices.  

Let $\overrightarrow F$ be an orientation of (each component of) $\bar F$ as described in Lemma \ref{orient}.  For each non-sink vertex $u$ of $\overrightarrow F$, pick an arbitrary outgoing edge $uv$, let $s$ be the social vertex that subdivides $uv$ and add edge $us$ to the matching $M$.  By construction, the added edges guarantee that $M$ is a matching.  Since every component of $F$ has at least 2 vertices, and there can be at most one non-sink per component of $\overrightarrow F$, we add an edge to $M$ for each of at least half the vertices of $\bar F$.  That is, we add at least $\frac{t_2}{6(r-1)}$ edges to $M$.  Therefore\\ $|M| \ge t_1/2+ \frac{t_2}{6(r-1)} \ge \frac{qcn}{6(r-1)}\geq \frac{2r-2}{2r-1}\frac{cn}{6(r-1)} \geq \frac{2r-1}{3(4r-1)(2r-1)}n = \frac{n}{12r - 3}.$
\qed \end{proof}

\section{Comparison to Bodlaender's Algorithm}
Bodlaender's algorithm~\cite{Bodlaender93,BDDFLP16}, in time $O(k^{O(1)}n)$, finds (i) a maximal matching, say $M$, of size at least $\frac{n}{k^6}$ or (ii) a set of simplicial vertices, say $X$, of size at least $\frac{n}{k^6}$ or decides that the treewidth of $G$ is at least $k + 1$. Then, the algorithm either removes $X$ if its size is at least $\frac{n}{k^6}$ or contracts $M$ it its size is at least $\frac{n}{k^6}$ from the graph and recursively finds the optimal tree decomposition of the resulting graph. Thus, in each recursive step, the size of the graph reduced by $1 - \frac{1}{k^6}$ factor. Our algorithm, by contrast, guarantees that one can always find a set of $r$-families (that can be handled easily) or a maximal matching of size $\frac{n}{37}$ (Theorem~\ref{noNiceK23}) so that the contraction reduces the size of the graph by a constant factor.
The full dependence of Bodlaendar's algorithm on $k$ is not computed in his work, so we do not attempt to find an exact comparison of the dependence in this abstract.
\end{document}